\documentclass[aps,epsfig,graphics,floatfix,mathbbm,a4paper,nofootinbib,prl]{revtex4}

\usepackage{amsmath,amsthm,amsfonts,amssymb,graphics,graphicx,epsfig,color,times,bbm}
\usepackage{enumerate}
\usepackage[thinlines,thiklines]{easybmat}
\usepackage[normalem]{ulem}

\begin{document}

\newcommand{\spec}{\rm{spec}}
\newcommand{\linearspan}{\rm{span}}
\newcommand{\id}{{\rm{id}}}
\newcommand{\R}{\mathbbm{R}}
\newcommand{\cB}{{\cal B}}
\newcommand{\cS}{{\cal S}}
\newcommand{\CP}{{\mathbb{C}}{\bf P}}
\newcommand{\cH}{{\cal H}}
\newcommand{\cN}{{\cal N}}
\newcommand{\rr}{\mathbbm{R}}
\newcommand{\cT}{{\cal T}}
\newcommand{\cC}{{\cal C}}
\newcommand{\cK}{{\cal K}}
\newcommand{\supp}{{\rm supp}}
\newcommand{\cV}{{\cal V}}
\newcommand{\cL}{{\cal L}}
\newcommand{\E}{{\cal E}}
\newcommand{\cX}{{\cal X}}
\newcommand{\V}{{\cal V}}
\newcommand{\cc}{{\cal{C}}}
\newcommand{\ii}{\mathbbm{1}}
\newcommand{\cM}{{\mathcal M}}
\newcommand{\ra}{\rightarrow}
\newcommand{\C}{\mathbb{C}}
\newcommand{\1}{\mathbbm{1}}
\newcommand{\F}{\mathbbm{F}}
\newcommand{\h}{\frak{h}}
\newcommand{\osc}{{\rm osc}}
\newcommand{\tr}[1]{{\rm tr}\left[#1\right]}
\newcommand{\gr}[1]{\boldsymbol{#1}}
\def\>{{\rangle}}
\def\<{{\langle}}
\newcommand{\be}{\begin{equation}}
\newcommand{\ee}{\end{equation}}
\newcommand{\bea}{\begin{eqnarray}}
\newcommand{\eea}{\end{eqnarray}}
\newcommand{\ba}{\begin{array}}
\newcommand{\ea}{\end{array}}
\newcommand{\ket}[1]{|#1\rangle}
\newcommand{\bra}[1]{\langle#1|}
\newcommand{\avr}[1]{\langle#1\rangle}
\newcommand{\red}[1]{{\bf \textcolor{red}{[{#1}]}}}
\newcommand{\blue}[1]{{\bf \textcolor{blue}{[{#1}]}}}
\newcommand{\green}[1]{{\bf \textcolor{green}{[{#1}]}}}
\newcommand{\D}{{\cal D}}
\newcommand{\eq}[1]{Eq.~(\ref{#1})}
\newcommand{\ineq}[1]{Ineq.~(\ref{#1})}
\newcommand{\sirsection}[1]{\section{\large \sf \textbf{#1}}}
\newcommand{\sirsubsection}[1]{\subsection{\normalsize \sf \textbf{#1}}}
\newcommand{\ack}{\subsection*{\normalsize \sf \textbf{Acknowledgments}}}
\newcommand{\front}[5]{\title{\sf \textbf{\Large #1}}
\author{#2 \vspace*{.4cm}\\
\footnotesize #3}
\date{\footnotesize \sf \begin{quote}
\hspace*{.2cm}#4 \end{quote} #5} \maketitle}
\newcommand{\eg}{\emph{e.g.}~}

\newcommand{\rational}{\mathbb Q} %rational
\newcommand{\real}{\mathbb R} %real
\newcommand{\complex}{\mathbb C} %complex
\renewcommand{\nat}{\mathbb N} %natural
\newcommand{\integer}{\mathbb Z} %integer

\newcommand{\proofend}{\hfill\fbox\\\medskip }

%linear span
\newcommand{\As}{\mathfrak{S}(\Ao)} % linear span
\newcommand{\ABs}{\mathfrak{S}(\Ao,\Bo,\ldots)} % linear span
\newcommand{\Ascomp}{\As^\perp} %complement
\newcommand{\Bs}{\mathcal{B}} % linear span
\newcommand{\Bscomp}{\mathcal{B}^\perp} %complement
\newcommand{\os}{\mathfrak{S}} % operator system

%other
\newcommand{\rank}[1]{\mathrm{rank}(#1)} %rank

\definecolor{gray}{rgb}{.6,.6,.6}
\DeclareRobustCommand{\cut}[1]{{\color{gray}\sout{#1}}}

%---------------------------------------------------------------------------

\theoremstyle{plain}
\newtheorem{theorem}{Theorem}
\newtheorem{proposition}[theorem]{Proposition}
\newtheorem{lemma}[theorem]{Lemma}
\newtheorem{definition}[theorem]{Definition}
\newtheorem{corollary}[theorem]{Corollary}
\newtheorem{conjecture}[theorem]{Conjecture}

\theoremstyle{definition}
\newtheorem{example}[theorem]{Example}
\newtheorem{remark}[theorem]{Remark}

\title{{Standard super-activation for Gaussian channels requires squeezing}}
\author{Daniel Lercher$^{1,}$\footnote{Electronic address: daniel.lercher@ma.tum.de}}
\author{G\'eza Giedke$^{1,2,}$\footnote{Electronic address: giedke@mpq.mpg.de}}
\author{Michael M.~Wolf$^{\, 1,}$\footnote{Electronic address: m.wolf@tum.de}}
\affiliation{$^1$Department of Mathematics, Technische Universit\"at M\"unchen, 85748 Garching, Germany\\$^2$Max-Planck-Institut f\"ur Quantenoptik, 85748 Garching, Germany}
\date{\today}

\begin{abstract}The quantum capacity of bosonic Gaussian quantum channels
can be  non-additive in a particularly striking way:
a pair of such optical-fiber type channels can individually have zero
quantum capacity but super-activate each other such that the combined
channel has strictly positive capacity. This has been shown in [Nature
Photonics 5, 624 (2011)] where it was conjectured that squeezing
is a necessary resource for this phenomenon. We provide a proof of this conjecture by showing that for gauge covariant channels a Choi matrix with positive partial transpose implies that the
channel is entanglement-breaking. In addition, we construct an example which shows that this implication fails to hold for Gaussian channels which arise from passive interactions with a squeezed
environment. \end{abstract}

\maketitle

\section{Introduction}\label{SectIntro}

A question at the heart of information theory---classical as well as
quantum---is, how to transmit information reliably, given imperfect
resources. The means to transmit information are referred to as the (typically
noisy) \emph{channel}. Its most important quantitative characteristic is how many
units of information can be reliably sent per use in the limit of many uses of
the channel. This number is called the \emph{capacity} of the channel
\cite{Cover1991}.

For classical memoryless channels Shannon posed and answered this questions in
his groundbreaking work \cite{Shannon48} in which he provided a tractable
formula for the capacity of any such channel.

One of the fundamental insights of quantum information theory is that for
quantum channels several distinct capacities can be defined depending on which
kind of information (classical, quantum, \dots) is to be sent and which
ancillary resources are provided \cite{Wilde2013}. Moreover, despite
considerable progress, no closed general expression is known for the classical
or quantum capacity. What complicates matters in the quantum world is that
quantum correlations between different channel uses can improve the
performance or, mathematically speaking, lead to non-additivity effects. 
This
has two practical consequences: not only is it in general necessary to
entangle the channel inputs over many uses to fully exploit the capacity of
the channel, this capacity can in some cases be further enhanced by combining
two different channels, so that their joint capacity exceeds the sum of the
two parts.
One of the most striking examples of these effects is the
\emph{super-activation} of the quantum capacity: a pair of channels can
individually have zero quantum capacity, but when combined give rise to a
channel whose quantum capacity is strictly positive \cite{SmYa08}.

In \cite{SSY11} it was shown that this effect can 
even occur within the
practically most important class of (bosonic) Gaussian channels. Among others,
they describe the transmission of the continuous degrees of freedom of light
in free space and in optical fibers 
\cite{Weedbrook2012} (taking into account
the most common loss and noise mechanisms) as well as the time evolution of
quantum memories which are based on collective excitations 
in atomic systems
\cite{HSP10}.  

One of the channels used in the construction in \cite{SSY11} can indeed be
regarded as a simple model of a lossy single-mode optical fiber. The second
channel, however, involves 
\emph{squeezing}, an experimentally more demanding
interaction arising from processes in which, e.g., photons are created or
annihilated pairwise \cite{Dod02}. Squeezing can be produced by birefringent materials and
selectively reduces the quantum noise associated with certain observables of
the electromagnetic field below its standard quantum value. In \cite{SSY11}
super-activation was shown employing
a high degree of squeezing within the
interaction between a two-mode system and its environment---something
considerably more difficult to realize than simple loss processes. The authors
of \cite{SSY11} write: ``\emph{Although an example using only linear optical
  elements would be desirable, we suspect, but cannot prove, that none
  exist.}''  The present paper aims at settling this conjecture in the
affirmative.

Currently, there is basically one approach towards super-activation. This is
based on the fact that there are only two classes of channels known, which
have provably zero quantum capacity: channels with a symmetric extension and
so-called PPT channels. Since both classes are closed with respect to parallel
composition, the only combination with a chance of successful super-activation
is to take one channel from each class. In this work we show that if we
restrict ourselves to passive Gaussian channels (i.e., those not involving
squeezing), then the set of PPT channels becomes a strict subset of the set of
channels with a symmetric extension, therefore rendering super-activation
impossible.

In the following we will first briefly review the main mathematical tools to
describe Gaussian channels and precisely define ``passive channels'', i.e.,
the class of Gaussian channels without squeezing within which we then show
that super-activation is not possible.

\section{Prerequisites}
We begin with recalling  basic notions and results needed for our purpose.
\subparagraph{Gaussian states and channels:}\label{SSectionGStCh}
We consider a continuous variable system of $n$ bosonic modes whose description involves $n$ pairs of generalized position and momentum operators $Q_k, P_k$ which may correspond to the quadratures of
electromagnetic field modes. With the definition $R:=\left(Q_1,P_1,\dots,Q_{n},P_{n}\right)$ the canonical commutation relations read 
	\be\label{eqSigma}
		[R_k,R_l]=i\left(\sigma_n\right)_{kl}\1,
		\quad
	{\rm with}\;
	\sigma_n:=\bigoplus_{i=1}^n 
	\left(
	 \begin{array}{cc}
	 0 & 1\\
	-1& 0
	\end{array}
	\right)
	\ee
being the symplectic form.\\
We associate with every density operator $\rho$ its displacement vector $d$, with $d_k:=\tr{\rho R_k}$ and its covariance matrix $\Gamma$ with $\Gamma_{kl}:=\tr{\rho\{R_k-d_k\1,R_l-d_l\1\}_+}$,
$k,l=1,\ldots,2n$. $d$ and $\Gamma$ contain the first and second moments of the corresponding phase space distribution.\\
The significance of $d$ and $\Gamma$ becomes evident as we restrict our attention from now on to Gaussian states.
These are defined as
quantum states with a Gaussian Wigner phase space distribution function, see \cite{HolevoBook11}. 
The Hilbert space of a continuous variable system is of infinite dimensions. The restriction to the set of Gaussian states allows for a much simpler description that requires only
a finite number of parameters.
In particular a Gaussian state is completely specified by $d$ and $\Gamma$, the latter being
any real symmetric matrix that satisfies the uncertainty relation
	\be\label{eqHUR}
		\Gamma\geq i\sigma_n.
	\ee
In the following all states are assumed to be centered (i.e. $d=0$) since displacements in phase space are local unitaries in Hilbert space which are irrelevant for our purpose.
Gaussian states form only a small subset of the state space. Yet they provide a good description for many experimentally accessible states, including 
coherent laser beams and Gibbs states of electromagnetic modes.

Now we turn our attention to state transformations, which are
  mathematically described by completely positive maps. If such a map is also trace-preserving
and preserves the Gaussian nature of states, it is called \emph{Gaussian channel}, see \cite{EiWo07}. Again neglecting its effect on $d$ it can be characterized by
its action on covariance matrices, which is given by
	\be\label{eqGChMatrix}
		\Gamma\mapsto X\Gamma X^T +Y,\qquad X,Y=Y^T\in\cM_{2n}(\R).
	\ee
For a pair of real matrices $X$ and $Y=Y^T$ to describe a bona fide Gaussian channel it is necessary and sufficient that
	\be\label{eqCP}
		Y+i\left(\sigma_n - X\sigma_n X^T\right)\geq 0.
	\ee
Unitary Gaussian evolutions then correspond to $Y=0$ and $X$ being real symplectic, i.e. $X\in Sp(2n,\R)=\{S\in\cM_{2n}(\R)|\,S\sigma_n S^T=\sigma_n\}$.
As is well known, every channel can be realized by
a unitary dilation $U$ describing the evolution of the system coupled to
  an environment in state $\rho_E$
\be\label{eqGChannel}
 \rho\mapsto {\rm tr}_E\left[U(\rho_E\otimes\rho)U^\dagger\right].
\ee
For a Gaussian $n$-mode channel, $\rho_E$ can be chosen as a Gaussian state of $n_E\leq 2n$
environmental modes and $U$ as a $(n+n_E)$-mode Gaussian unitary.\cite{caruso2008multimode}. 
Gaussian unitaries are generated by Hamiltonians that are quadratic in the
generalized position and momentum operators $R_k$ and represent a family of
transformations that can be realized, e.g., in quantum optical experiments
\cite{BrLo05}.

On the level of phase space, the evolution of the covariance matrices of
  environment and input state, $\Gamma_E$ and $\Gamma$, is governed by the
  symplectic transformation $S$ and looks like \be \left(
  \Gamma_E\oplus\Gamma \right) \mapsto S \left( \Gamma_E\oplus\Gamma \right)
S^T.  \ee In the notation $ S=\left(
\begin{array}{cc}
S_1&S_2\\
S_3&S_4
\end{array}
\right)
$
that reflects the decomposition of the total system into the environment and
the $n$-modes system, one finds $X=S_4$ and $Y=S_3\Gamma_E S_3^T$. Here,
  $S$ describes the most general Gaussian coupling between
  system and environment including, in particular, squeezing. Since our aim in
the following is to analyse capabilities of quantum channels in the absence of
this (in practice very demanding) ingredient, we now proceed to take a closer
look at the simpler set of unitaries generated by \emph{passive}
Hamiltonians. 

\subparagraph{The Symplectic Orthogonal Group:} 
Passive Hamiltonians are those given by quadratic expressions in $Q_k,P_k$ or,
equivalently, in the annihilation and creation operators
$a_k:=(Q_k+iP_k)/\sqrt{2}$ and $a_k^\dag$ that commute with the total particle
number operator $\sum_ka_k^\dag a_k$. They take the form 
\be\label{eqPassGCh}
\quad H= \sum_{k,l=1}^{m} h_{kl} a^\dagger_k a_l + h.c.
\ee
with $h_{kl}\in\mathbb{C}$. 
A unitary Gaussian evolution
generated by a passive Hamiltonian as in \eqref{eqPassGCh} corresponds to a
symplectic orthogonal matrix $S\in K(m):=Sp(2m,\R)\cap
O(2m,\R)$. Mathematically, $K(m)$ is the largest compact subgroup of the real
symplectic group. Physically, it corresponds to the set of operations which
can be implemented using beam splitters and phase shifters only \cite{ReZe94}.
Note that some of the most common channels such as the lossy channel and the
thermal noise channel can be described by passive Gaussian dilations.

In the following, it is useful to exploit that the group $K(m)$ is
isomorphic to the group $U(m)$ of $m\times m$ unitary matrices
\cite{ADMS95b}. This can be verified easily: First one observes that elements
$R\in\cM_{2m}(\R)$ in the commutant of $\sigma_m$ have the form
\begin{equation}
\label{eqCommSigma}
[\sigma_m,R]=0\Leftrightarrow R=(r_{ij})_{i,j=1}^m,\quad\textnormal{with}\
r_{ij}=
\left(
\begin{array}{cc}
a_{ij}&b_{ij}\\
-b_{ij}&a_{ij}
\end{array}
\right),\;a_{ij},b_{ij}\in\R.
\end{equation}
With this result one verifies that the map  
\begin{align}
\label{eqGroupIso}\Lambda:U(m)&\rightarrow K(m)\\
\nonumber (c_{ij})&\mapsto (C_{ij}),\quad C_{ij}=
\left(
\begin{array}{cc}
\Re{(c_{ij})}&\Im{(c_{ij})}\\
-\Im{(c_{ij})}&\Re{(c_{ij})}
\end{array}
\right)
\end{align}
is indeed a group isomorphism.

At this point we add two observations that will help us later  to exploit the particular structure of real symplectic orthogonals \eqref{eqGroupIso}. The set
\be
\mathcal{C}_n:=\left\{
\left(
\begin{array}{cc}
A&B\\
-B&A
\end{array}
\right)\middle| 
%\ba{c}
A,B\in%\\
\cM_n(\R)
%\ea
 \right\}
\ee
together with the operation of matrix multiplication forms a semigroup with neutral element. As such, it is  isomorphic to $\cM_n(\mathbb{C})$. An  isomorphism  is given by
\begin{equation}
\label{eqMonoidIso}
\left(
\begin{array}{cc}
A&B\\
-B&A
\end{array}
\right)
\mapsto A+iB\,.
\end{equation}
And finally, for complex square matrices $A, B \in\cM_n(\C)$ one finds the following criterion for positive-semidefiniteness \cite{RealComplPosDef}
\be\label{eqPosDefEquival}
\left(
\begin{array}{cc}
A&B\\
-B&A
\end{array}
\right)\geq 0 \Leftrightarrow A\pm i B\geq 0\,.
\ee

Now we combine the passive Gaussian unitaries with

\subparagraph{Properties of Gaussian Channels:}\label{SSecPropGCh}

\begin{definition}\label{DefPassGener} We call a Gaussian channel ``passive''
  if it can be generated by a $m=(n_E+n)$-mode passive Hamiltonian $H$ \eqref{eqPassGCh} that
  couples the system to an environment in a Gibbs state $\rho_E$ of a passive
  Hamiltonian $H'$:   
\be\label{eqGibbsEnv}
\rho_E=\frac{{\rm e}^{-\beta H'}}{\tr{{\rm e}^{-\beta H'}}},\quad H'= \sum_{k,l=1}^{n_E} h'_{kl} a^\dagger_k a_l + h.c.
\ee
\end{definition}
One can show that, as a consequence of the normal mode decomposition of
Gaussian states, \eqref{eqGibbsEnv} is equivalent to
$\left[\Gamma_E,\sigma_{n_E}\right]=0$, where $\Gamma_E$ is the covariance 
matrix of the Gaussian state $\rho_E$ and $\sigma_{n_E}$ is the corresponding
symplectic form. This implies for passive Gaussian channels 
\be\label{eqYcommSigma}
\left[Y,\sigma_n\right]=\left[S_3\Gamma_E S_3^T,\sigma_{n}\right]=0,
\ee
as one derives from the block structure of $S_3$ \eqref{eqGroupIso} and the
analogous structure of the elements in the commutant of $\sigma_n$
\eqref{eqCommSigma}. Similarly we find 
$\left[ X,\sigma_n\right]=\left[S_4,\sigma_n\right]=0$. This is a useful property;
  it implies that the channel commutes with the passive unitary generated by
  the number operator (represented on phase space by
  $e^{i\phi\sigma_n}$). Gaussian channels with this property are called
  ``gauge covariant'' \cite{Hol10} and are characterized by matrices $(X,Y)$
  commuting with $\sigma_n$.

% \begin{definition}\label{DefGaugeCov}
% In accordance with \cite{Hol10} we call a Gaussian channel characterized by
% $(X,Y)$ ``gauge covariant'', if
% $\left[X,\sigma_n\right]=\left[Y,\sigma_n\right]=0$. 
% \end{definition}

Accordingly, all passive Gaussian channels are gauge covariant. 
Let us add a remark on our definition of passive channels. 
Clearly, no active (squeezing-type) interaction is needed to generate
interaction or environmental state. Coupling the system to a Gibbs state
for a passive Hamiltonian includes the typical
situation, since   both non-equilibrium states and squeezing Hamiltonians
usually require active preparation, which do not naturally occur in the
uncontrolled environment. \footnote{Another plausible definition of passive
  channels would be to require passive 
coupling to an environment in a state that is
not squeezed (i.e., $\Gamma_E$ has no eigenvalue smaller than 1). This
includes our definition but adds high-temperature Gibbs states of squeezing
Hamiltonians. We cannot rule out super-activation for these types of channels.
}

Before stating and proving our main result, we finally need to
  characterize on the level of covariance matrices the two types of channels
  that have been used in the examples of super-activation. There, two types of
  noisy channels with restricted capability to transmit entanglement have been
  studied. On the one hand, \emph{entanglement-breaking channels} that, when
  acting on part of an entangled state always produce a separable output and
  \emph{PPT channels} (for ``positive partial transpose''), that may transmit
  entanglement but only in the form of states that remain positive under
  partial transposition and thus represent bound entanglement that cannot be
  locally distilled into pure entangled states \cite{BDSW96,HHH98}. Since a
  finite quantum capacity requires the ability to transmit (asymptotically)
  pure entangled states, this capacity vanishes for PPT channels.

A Gaussian channel is entanglement-breaking if and only if $Y$ admits a
decomposition into real matrices $M$ and $N$ such that \cite{Holevo08}
\be\label{eqEB} Y=M+N,\quad M\geq i\sigma_n,\quad N\geq iX\sigma_n X^T.  \ee
This reflects the fact that any entanglement-breaking quantum channel consists
of a measurement, followed by a state preparation depending on the outcome of
the measurement \cite{shor03eb, holwer05infdimsepeb}. As a consequence, every
entanglement-breaking channel has a symmetric extension and therefore
  zero quantum capacity by the no-cloning lemma \cite{WoZu82}.

A quantum channel $T$ is called a PPT channel if $\theta\circ T$ is completely
positive, where $\theta$ denotes time reversal, which in Schr\"odinger
representation corresponds to transposition \cite{PPTchannel}. A
Gaussian channel characterized by $\left( X,Y\right)$ is PPT iff
\be\label{eqPPT} 
Y-i\left(\sigma_n + X\sigma_n X^T\right)\geq 0.  
\ee

Now we are in a position to turn to the question whether superactivation of
PPT and entanglement-breaking channels is possible in the Gaussian setting
without squeezing. The main technical step is to show that for gauge covariant
channels being PPT and being entanglement-breaking are equivalent.

\section{Main result}\label{SectResults}
\begin{proposition}\label{PropPPTimpliesEB}
A gauge covariant Gaussian channel T is entanglement-breaking iff it is PPT.
\end{proposition}
\begin{proof}
Evidently, entanglement-breaking implies PPT, so we have to prove the reverse implication.
To this end it is convenient to reorder the canonical coordinates as $\left(Q_1,\dots,Q_{n}, P_1,\dots , P_n \right)$. In this notation
\be
	\sigma_n=
	\left(
	\begin{array}{cc}
	&\1_n\\
	-\1_n&
	\end{array}
	\right)\;
	{\rm and}\; 
	Y=
	\left(
	\begin{array}{cc}
	Y_1&Y_2\\
	-Y_2&Y_1
	\end{array}
	\right).
\ee
The latter follows from \eqref{eqYcommSigma} together with \eqref{eqCommSigma}. Below we omit the index of $\sigma_n$.\\
We prove first that we can restrict ourselves to the case $X=\hat X\oplus\hat X$, where by virtue of the symplectic singular value decomposition \cite{MMWnotSoNormalMode} the matrix $\hat X$ is diagonal
and non-negative.
Assume, this is not the case. Then we can replace the unitary dilation $U$ associated with $T$, which describes the interaction between system and environment, by $ U'=(\1_E\otimes U_G) U (\1_E\otimes U_F)$. $U_F$ and $U_G$ are
passive unitary evolutions that act only on the system. They correspond to symplectic transformations $F, G\in K(n)$ in phase space. We denote the resulting channel by $T'$. $T$ is PPT iff $T'$ is PPT. The same
holds for the entanglement-breaking property. We find
\begin{equation}
\label{eqXdiag}
X'= GXF=
\left(
\begin{array}{cc}
G_1&G_2\\
-G_2&G_1
\end{array}
\right)
\left(
\begin{array}{cc}
X_1&X_2\\
-X_2&X_1
\end{array}
\right)
\left(
\begin{array}{cc}
F_1&F_2\\
-F_2&F_1
\end{array}
\right),
\end{equation}
with $F_1+i F_2, G_1+iG_2 \in U(n)$.
Now we can exploit the isomorphism \eqref{eqMonoidIso} and choose $F_{1,2}$ and $G_{1,2}$ such that $(G_1+iG_2)(X_1+iX_2)(F_1+iF_2)=:\hat X$ is the singular value decomposition of $X_1+iX_2$. Hence,
$X'=\hat X\oplus\hat X$ with $\hat X$ diagonal and non-negative. 

We will now exploit criterion \eqref{eqEB} by showing that the decomposition of $Y$  into $M:=Y-X^2$ and $N:=X^2$ obeys the required conditions, which read:
\bea
\label{eqEBdiagX}
N-iX\sigma X^T&=&X(\1-i\sigma) X^T\geq 0,\\
\nonumber
M-i\sigma&=&\left(
\begin{array}{cc}
Y_1-\hat X^2&Y_2-i\1\\
-Y_2+i\1&Y_1-\hat X^2
\end{array}
\right)
\geq 0
\eea
The first inequality in \eqref{eqEBdiagX} follows simply from $(\1-i\sigma)\geq 0$. In order to arrive at the second inequality we use \eqref{eqPosDefEquival} and rewrite the inequality as
\be
\begin{array}{rc}
Y_1-\hat X^2\pm i (Y_2-i\1)\geq 0&\Leftrightarrow\\ 
Y_1\pm i Y_2-(\hat X^2\mp \1)\geq 0 &{\Leftrightarrow}\\
Y_1\pm iY_2-(\hat X^2+\1)\geq 0&{\Leftrightarrow}\\
\left\{
\begin{array}{l}
Y_1+ iY_2+(\hat X^2+\1)\geq0\\
Y_1-iY_2-(\hat X^2+\1)\geq0
\end{array}
\right.
&\Leftrightarrow\\
Y_1\pm i\left(Y_2-i(\1+\hat X^2)\right)\geq0&.
\end{array}
\ee
Here we used two elementary facts: (i) a matrix is positive iff its complex conjugate is positive, and (ii) the sum of two positive matrices is again positive.
In the last line, with \eqref{eqPosDefEquival}, we recover the PPT criterion \eqref{eqPPT}
\begin{equation}
 Y-i(\sigma + X\sigma X^T)=Y-i\sigma(\1+X^2)
=\left(
\begin{array}{cc}
Y_1&Y_2-i(\1+\hat X^2)\\
-Y_2+i(\1+\hat X^2)&Y_1
\end{array}
\right)\geq 0\,,
\end{equation}
which concludes the proof.
\end{proof}
\begin{proposition}[No super-activation without squeezing]\label{CorrNoSuperact} Let $T_1$, $T_2$ be passive Gaussian quantum channels. If each channel either has a symmetric extension or satisfies the PPT
property, then $Q\left(T_1\otimes T_2\right)=0$.
\end{proposition}
\begin{proof}
Let $T_i\,(i=1\,{\rm or}\,2)$ be PPT. $T_i$ is gauge covariant, because it is passive, and according to Prop. \ref{PropPPTimpliesEB} it is thus entanglement-breaking. In particular, it has a symmetric extension,
which then also holds for $T_1\otimes T_2$.  Hence, the combined channel has zero quantum capacity.
\end{proof}
\section{Passive interactions with a squeezed environment}\label{SubSect2}
We now consider an example of a Gaussian channel $T$ for $n=n_E=2$. $T$ is generated by a passive interaction, as in \eqref{eqPassGCh}, but  the environment is assumed to be in a mixed squeezed state
$\rho_E$ (i.e. $\det\Gamma_E\neq 1$, $\Gamma_E\geq i\sigma_2$ and $\Gamma_E\ngeq\1_4)$.
$T$ will be shown to be PPT but not entanglement-breaking. We omit the index of $\1_2$ and $\sigma_2$ and choose
\bea
&\Gamma_E=\frac{3+\sqrt{13}}{2}
\left(
\begin{array}{cccc}
5 &&3 & \\
& 5&&-3\\
3 & &2&\\
& -3 &&2
\end{array}
\right),&\\
&S=\sqrt{\frac{1}{3}}
\left(
\begin{BMAT}(b){cc3}{cc3cc}
-\sqrt{2}\1 &\\
& -\1\\
\1 & \\
& \sqrt{2}\1  
\end{BMAT}
\begin{BMAT}(b){cc}{cc3cc}
\1 & \\
& \sqrt{2}\1 \\
\sqrt{2}\1 & \\
& \1
\end{BMAT}
\right).&
\eea
 $S$ represents two beamsplitters: one of transmittivity $\frac{2}{3}$ that couples the first system mode to the first mode of the environment and a second of transmittivity $\frac{1}{3}$ that acts between
the two remaining modes.
The corresponding $(X,Y)$, which characterize the Gaussian channel, then read
\begin{eqnarray}\label{XYexampPPTnotEB}
&X=\sqrt{\frac{1}{3}}
\left(
\begin{array}{cc}
\sqrt{2}\1 &\\
& \1\\
\end{array}
\right),&\\
&Y=\frac{3+\sqrt{13}}{6}
\left(
\begin{array}{cccc}
5 &&3\sqrt{2} & \\
& 5&& -3\sqrt{2} \\
3\sqrt{2} & &4 & \\
& -3\sqrt{2} && 4
\end{array}
\right).&
\end{eqnarray}
\begin{proposition}\label{PropPPTnotEB}
The Gaussian channel determined by \eqref{XYexampPPTnotEB} exhibits the PPT property but it is not entanglement-breaking. 
\end{proposition}
\begin{proof} Equations \eqref{eqCP} and \eqref{eqPPT} are satisfied as one verifies explicitly.\\
It remains to show that $T$ is not entanglement-breaking. With the inequalities \eqref{eqEB} in mind we observe that this is equivalent to
\begin{equation}
\label{eqMaxLambda}
 \max_{(\lambda,M)\in\mathcal{D}}
       \lambda < 1,\quad\mbox{where}\,
\mathcal{D} =
\left\{
\begin{array}{c}
(\lambda,M)\in\\
\R\times\cM_4(\R)
\end{array}
\middle|
\begin{array}{rl}
M&=M^T,\\
M&\geq \lambda i\sigma,\\
Y-M&\geq \lambda i X\sigma X^T
\end{array}
\right\}.
\end{equation}
This is a semi-definite program \cite{ConvexOpt}, so that the corresponding dual program can be used to construct a witness which certifies \eqref{eqMaxLambda}. Its specific form is given in the appendix.
\end{proof}
\section{Discussion}
Super-additivity of channel capacities is one of the surprising between
classical and quantum information theory and its mechanisms and quantitative
importance are still poorly understood.  Super-activation of the quantum
capacity is one of the most extreme examples of such effects.

In the practically relevant Gaussian setting, super-activation can be
achieved using squeezing, adding to a long list of quantum
effects -- entanglement generation, metrology, information coding or
continuous-variable key-distribution -- whose realization squeezing
enables. In the
Gaussian regime, we know that it is sometimes even a necessary 
resource. This is the case for entanglement generation \cite{WEP03} and, as 
we have proven in this work, as well for super-activation of the
quantum capacity. In the latter case, however, the proof of necessity
relies on the framework---the basic idea behind the construction of
all currently known instances of super-activation. In order to make a
stronger statement, we would need to know whether there are other
types of channels with zero quantum capacity \cite{SmSm12}.  

Another question, which suggests itself, is how much squeezing is
necessary within the given framework. Unfortunately, we do at the
moment not see an approach towards settling this quantitative
question.

\subparagraph{Acknowledgements:} 
DL and MMW acknowledge financial support from the  
CHIST-ERA/BMBF project CQC and the QCCC programme of the Elite Network of
Bavaria. GG acknowledges the project MALICIA and its financial support
  of the Future and Emerging Technologies (FET) programme within the Seventh
  Framework Programme for Research of the European Commission, under FET-Open grant
number: 265522.

\section{Appendix}
In the following we show how to certify \eqref{eqMaxLambda}.
Note that with the notation $
\tilde Y=
0_4\oplus
Y
$, $
\tilde X=
i\sigma
\oplus iX\sigma X^T
$, $
\tilde M=
M
\oplus
-M 
$
, the two inequalities in the definition of $\cal{D}$ can be rewritten as
\be\label{eqConstr}
\lambda \tilde X+\tilde Y+\tilde M \geq 0.
\ee
In the following we confirm \eqref{eqMaxLambda} by showing that for all $(\lambda,M)\in\mathcal{D}$, $\lambda \leq 0.94$. For this purpose, let us define the witness matrix $\Omega$,
\be\label{eqOmega}
\Omega=
(A+iB)\oplus(A+iC).
\ee
\bea
\nonumber
&A=
\left(
\begin{array}{cccc}
a_1 & &-a_3 & \\
& a_1&&  a_3\\
-a_3& &a_2 &  \\
&  a_3&& a_2
\end{array}
\right),
&a=\left(
\ba{c}
0.512\\
0.722\\
0.592
\ea
\right),
\\
\nonumber
&B=
\left(
\begin{array}{cccc}
 & b_1& & b_3\\
-b_1& &b_3& \\
& -b_3& & b_2 \\
-b_3& & -b_2& 
\end{array}
\right),
&b=\left(
\ba{c}
-0.212\\
0.552\\-0.368
\ea
\right),
\\
\nonumber
&C=
\left(
\begin{array}{cccc}
 & c_1& & c_3\\
-c_1& &c_3& \\
& -c_3&& c_2 \\
-c_3& &-c_2&  
\end{array}
\right),
&c=\left(
\ba{c}
0.39\\-0.3\\0.368
\ea
\right).
\eea
and state some of its properties:
\begin{enumerate}[(i)]
\item $\Omega$ is positive definite.
\item $\forall (\lambda,M)\in\mathcal{D}:\, \tr{\Omega\tilde M}= i\tr{(B-C)M}=0$, since $(B-C)$ is anti-symmetric and $M$ is symmetric.
\item $\tr{\Omega \tilde X}=2(b_1+b_2+\frac{2}{3}c_1+\frac{1}{3}c_2)=1$
\item $\tr{\Omega \tilde Y}=\left(1+\frac{\sqrt{13}}{3}\right)\left(5a_1+4a_2-6\sqrt{2}a_3\right)<0.94$
\end{enumerate}
Let now be $(\lambda,M)\in\mathcal{D}$. Applying (ii) and (iii) in the first line and (i) together with \eqref{eqConstr} in the third line leads to
\bea\nonumber
\lambda-\tr{\Omega\tilde Y}	&=&-\lambda\tr{\Omega\tilde X}-\tr{\Omega\tilde Y}-\tr{\Omega\tilde M}\\
					&=&-\tr{\Omega\left(\lambda\tilde X +\tilde Y +\tilde M\right)}\\ \nonumber
					&\leq& 0.
\eea
With (iv) we obtain $\lambda<0.94.$
\bibliographystyle{unsrt}

\begin{thebibliography}{10}
\bibitem{RealComplPosDef}
This equivalence can be seen from the similarity transformation $ \left(
  \begin{array}{cc} A-iB&\\ &A+iB \end{array} \right) =$\\ $ \frac{1}{\sqrt{2}}
  \left( \begin{array}{cc} \1&i\1\\ i\1&\1 \end{array} \right) \left(
  \begin{array}{cc} A&B\\ -B&A \end{array} \right) \frac{1}{\sqrt{2}} \left(
  \begin{array}{cc} \1&-i\1\\ -i\1&\1 \end{array} \right) $.
\end{thebibliography}

\end{document}